\documentclass[a4paper,reqno]{amsart} 
\usepackage{amsthm,amsmath,amssymb}
\usepackage{MnSymbol,bbm} 
\usepackage[T1]{fontenc}
\usepackage[normalem]{ulem}
\usepackage{colortbl}
\usepackage{booktabs,wrapfig}         
\usepackage[linktocpage=true]{hyperref}

\usepackage{color} 

\newcounter{mnotecount}[section]

\makeatletter
  \def\moverlay{\mathpalette\mov@rlay}
  \def\mov@rlay#1#2{\leavevmode\vtop{%
     \baselineskip\z@skip \lineskiplimit-\maxdimen
     \ialign{\hfil$#1##$\hfil\cr#2\crcr}}}
\makeatother

\DeclareMathOperator*{\tho}{\text{\th}}
\DeclareMathOperator*{\edt}{\text{\dh}}

\newcommand{\NPl}{l}
\newcommand{\NPn}{n}
\newcommand{\NPm}{m}
\newcommand{\NPmbar}{\bar m}

\newcommand{\ba}[1]{\overline{#1}}
\renewcommand{\i}{\mathrm{i}}
\renewcommand{\d}{\mathrm{d}}
\newcommand{\hide}[1]{}
\newcommand{\half}{\frac{1}{2}}
\newcommand{\thalf}{\tfrac{1}{2}}
\newcommand{\FF}{\mathcal{F}}
\newcommand{\pert}{\epsilon}
\newcommand{\lie}{\mathcal{L}}
\newcommand{\Reals}{\mathbb R}
\newcommand{\PP}{\mathcal P}

\renewcommand{\Re}{\mathrm{Re}\, }
\renewcommand{\Im}{\mathrm{Im}\, }

\newcommand{\sa}{{A}}
\renewcommand{\sb}{{B}}
\renewcommand{\sc}{{C}}
\newcommand{\sd}{{D}}

\newcommand{\Y}[3]{\, {}_{#1}\hspace{-1.5pt}Y\hspace{-3pt}_{#2 #3}}

\newcommand{\eps}{{\epsilon}}
\newcommand{\io}{{\iota}}

\definecolor{g1}{rgb}{0.95,0.95,0.95}
\definecolor{g2}{rgb}{0.9,0.9,0.9}
\definecolor{g3}{rgb}{0.8,0.8,0.8}

\newcommand{\bTheta}{\mathbf \Theta}

\newcommand{\bQ}{\mathbf Q}

\newcommand{\massADM}{M}		
\newcommand{\angmom}{a}			
\newcommand{\MM}{\mathcal M}		
\newcommand{\linmass}{\mathcal{\dot M}}	

\theoremstyle{plain}
\newtheorem{thm}{Theorem}[section]

\newtheorem{prop}[thm]{Proposition}
\newtheorem{remark}[thm]{Remark}

\title{Charges for linearized gravity}

\author[S. Aksteiner]{Steffen Aksteiner}
\email{steffen@zarm.uni-bremen.de}

\address{QUEST, Leibniz University Hannover, Welfengarten 1, D-30167 Hannover, Germany}

\address{ZARM, University of Bremen, Am Fallturm 1, D-28359 Bremen, Germany}

\author[L. Andersson]{Lars Andersson}
\email{laan@aei.mpg.de}
\address{Albert Einstein Institute, Am M\"uhlenberg 1, D-14476 Potsdam,
  Germany}

\usepackage{comment}
\excludecomment{longversion}
\excludecomment{wrongversion}

\begin{document}

\date{January 12, 2013}

\begin{abstract}
Maxwell test fields as well as solutions of linearized gravity on the Kerr
exterior admit non-radiating modes, i.e. non-trivial time-independent
solutions.  
These are closely
related to conserved charges. In this paper we discuss the
non-radiating modes for linearized gravity, which may be seen to
correspond to the Poincare Lie-algebra. The 2-dimensional isometry group of
Kerr corresponds to a 2-parameter family of gauge-invariant non-radiating
modes representing infinitesimal perturbations of mass and azimuthal
angular momentum. We
calculate the linearized mass charge in terms of linearized Newman-Penrose
scalars.  
\end{abstract}

\maketitle

\section{Introduction}
The black hole stability problem, i.e. the problem of proving dynamical
stability for the Kerr family of black hole spacetimes, is one of the central
open problems in General Relativity. The analysis of linear test fields on
the exterior Kerr spacetime is an important step towards the full non-linear
stability problem. For test fields of spin 0, i.e. solutions of
the wave equation $\nabla^a \nabla_a \psi = 0$, estimates proving boundedness
and decay in time are known to hold. See 
\cite{finster:etal:2006:MR2215614,dafermos:rodnianski:2010arXiv1010.5137D,andersson:blue:kerrwave,tataru:tohaneanu:2011}
for references and background. 

The field equations for linear test fields of spins 1 and 2 are the 
Maxwell and linearized gravity\footnote{Note that linearized gravity is distinct from the massless spin-2 equation. On a type D background, any solution to the massless spin-2 equation is proportianal to the Weyl tensor of the spacetime. This fact is
referred to as the Buchdahl constraint, cf. \cite{Buchdahl:1958xv}, see also
equation (5.8.2) in \cite{PR:I}.} equations, respectively. 
These equations imply wave equations for the Newman-Penrose Maxwell and
linearized Weyl scalars. 
In particular, the Newman-Penrose scalars of spin weight zero
satisfy (assuming a suitable gauge condition for the case of linearized
gravity) analogs of the Regge-Wheeler equation.
These wave equations take the form 
$$
(\nabla^a \nabla_a + c_s \Psi_2) \psi_s = 0
$$
where for spin $s=1$, $c_1 = 2$, $\psi_1 = \Psi_2^{-1/3} \phi_1$, while for
spin $s=2$, $c_2 = 8$,  
and $\psi_2 = \Psi_2^{-2/3} \dot \Psi_2$. Here $\dot \Psi_2$ is the
linearized Weyl scalar of spin weight zero. See
\cite{aksteiner:andersson:2011} for details. 
As these scalars can 
be used
as potentials for the Maxwell and linearized Weyl fields, 
one may apply the
techniques developed in the previously mentioned papers to prove estimates
also for the Maxwell and linearized gravity equations. This approach has been
applied in the case of the Maxwell field on the Schwarzschild background in
\cite{blue:maxwell}. 
 
In contrast to the spin-0 case, the spin 1 and 2 field equations on the Kerr exterior admit non-trivial finite energy time-independent solutions.
We shall refer to time-independent solutions as 
non-radiating modes. There
is a close relation between gauge-invariant 
non-radiating modes and conserved charge integrals. For the Maxwell field, 
there is a two-parameter family of  
non-radiating, Coulomb type solutions which carry the two
conserved electric and magnetic charges. 
In fact, 
a Maxwell field on the Kerr exterior will disperse exactly when it has
vanishing charges. For linearized gravity, however, there
are both 
non-radiating modes corresponding to 
gauge-invariant 
conserved
charges, and ``pure gauge'' non-radiating modes. 
Thus conditions ensuring
that a solution of
linearized gravity will disperse must be a combination of charge-vanishing
and gauge conditions. 

From the discussion above, it is clear that in order to prove 
boundedness and decay for higher spin test
fields on the Kerr exterior,
it is a necessary step to eliminate the non-radiating modes. Due in part to
this additional difficulty, decay estimates for the higher spin fields have
been proved only for Maxwell test fields. See \cite{blue:maxwell} for the
Schwarzschild case and \cite{andersson:blue:nicolas:maxwell} for the Kerr
case. 
In view of
the just mentioned relation between non-radiating modes and charges, 
an essential step in doing so involves setting conserved charges to zero. 
In order to make
effective use of such charge vanishing conditions, it is necessary to have
simple expressions for the charge integrals in terms of the field strengths. 
The main result
of this paper is to provide an expression for the conserved charge corresponding to the
linearized mass, in terms of linearized curvature quantities on the Kerr
background. 

We start by discussing the relation between charges and non-radiating modes
for the case of the Maxwell field. 
Let the symmetric valence-2 spinor 
$\phi_{AB}$ be the Maxwell spinor\footnote{The following discussion is in terms of the 2-spinor formalism,
cf. \cite{PR:I, PR:II}}, 
i.e. a
solution of the massless spin-1 (source-free Maxwell) equation 
$$
\nabla_{A'}{}^A \phi_{AB} = 0
$$
and let $\FF_{ab} = \phi_{AB} \eps_{A'B'}$ be the corresponding complex
self-dual two-form. The Maxwell equation takes the form $d\FF = 0$ and hence the charge integral 
$$
\int_{S} \FF
$$
depends only on the homology class of the surface $S$. Here real and imaginary parts correspond to electric and magnetic charges,
respectively. The Kerr exterior, being diffeomorphic to
$\Reals^4$ with a solid cylinder removed, contains
topologically non-trivial 2-spheres, and hence the Maxwell equation on the
Kerr exterior admits solutions with non-vanishing charges. In view of the
fact that the charges are conserved, it is natural that there is a
time-independent solution which ``carries'' the charge. In Boyer-Lindquist
coordinates, this takes the
explicit form 
\begin{equation}\label{eq:coloumb}
\phi_{AB} = \frac{c}{(r-ia\cos\theta)^2} \iota_{(A} o_{B)} \, ,
\end{equation}
where $c$ is a complex number, and $\iota_A, o_A$ are principal spinors for
Kerr. 

In order to prove boundedness and decay for the Maxwell field, it is
necessary to make use of the above mentioned facts, see
\cite{andersson:blue:nicolas:maxwell}. 
In particular, one eliminates the non-radiating modes by imposing the charge
vanishing condition 
\begin{equation}\label{eq:Fchargezero}
\int_{S} \FF = 0 \, .
\end{equation}
Written in terms of the Newman-Penrose scalars $\phi_I$, $I = 0,1,2$, the charge vanishing condition \eqref{eq:Fchargezero} in the Carter tetrad \cite{Znajek:1977} takes the form \cite{andersson:blue:nicolas:maxwell} 
\begin{equation}\label{eq:integrability}
\int_{S^2(t,r)}  2 V_L^{-1/2} \phi_1 
+ ia\sin\theta (\phi_0 - \phi_2 ) \d\mu = 0 \, ,
\end{equation} 
where $S^2(t,r)$ is a sphere of constant $t,r$ in the Boyer-Lindquist
coordinates, $V_L = \Delta/(r^2+a^2)^2$ and $\d\mu=\sin\theta\d\theta\d\varphi$. 
This yields a relation between the $\ell = 0, m=0$ spherical harmonic of 
$\phi_1$ and the $\ell=1, m=0$ spherical harmonics with spin weights $1$, $-1$
of $\phi_0$, $\phi_2$, respectively. 

Next, we consider the spin-2 case. Recall that the Kerr spacetime is a vacuum space of Petrov type D and hence, in addition to the Killing vector fields $\partial_t, \partial_\phi$ admits a ``hidden symmetry'' manifested by the existence of the 
valence-2 Killing spinor $\kappa_{AB} = \psi \, \iota_{(A} o_{B)}$. 
Here the scalar $\psi$ is determined up to a constant, 
which we fix by setting\footnote{This choice has the natural (non vanishing) Minkowski limit $\psi = r$.} $\massADM\psi^{-3} = -\Psi_2$ on a Kerr background.
In this situation, one may consider the spin-lowered version 
$$
\psi_{ABCD} \kappa^{CD} 
$$ 
of the Weyl spinor, which is again a massless spin-1 field and hence the complex
self-dual two-form
$$
\MM_{ab} = \psi_{ABCD}\kappa^{CD} \eps_{A'B'}
$$ 
satisfies the Maxwell equations $d\MM = 0$. 
The charge for this field defined on any topologically non-trivial 2-sphere
in the Kerr exterior is
\begin{equation}\label{eq:mass-charge}
\frac{1}{4\pi i} \int_{S} \MM = \massADM \, ,
\end{equation}
cf. \cite{Jezierski:Lukasik:2006} for a tensorial version (the calculation
has been done much earlier in \cite{jordan:ehlers:sachs:1961:II}, but not in
the context of Killing spinors and 
spin-lowering). 
Here $\massADM$ is the ADM mass \cite{ADM:1962:MR0143629}
of the Kerr spacetime\footnote{Equivalently, the mass parameter in the Boyer-Lindquist form of the Kerr line element.}. The relation between the mass and charge for 
the spin-lowered Weyl tensor $\MM$ is natural in view of the fact that the divergence 
$$
\xi^{A'A} = \nabla^{A'}{}_B \kappa^{AB}
$$
is proportional to $\partial_t$, see the discussion in \cite[Chapter 6]{PR:II}. 

Note that the charge \eqref{eq:mass-charge} is in general complex. The
imaginary part corresponds to the NUT charge, which is the gravitational
analog of a magnetic charge. 
Details are not discussed in this paper, see
\cite{Ramaswamy:Sen:1981} for the construction of charge integrals in NUT
spacetime. 

For linearized gravity on the Kerr background, the non-radiating modes include
perturbations within the Kerr family, i.e. infinitesmal changes of mass and
axial rotation speed. We denote the parameters for these 
deformations $\dot \massADM, \dot \angmom$. Since $\massADM, \angmom$ are
gauge-invariant quantities, it is not possible to eliminate these modes by 
imposing a gauge condition. 
A canonical
analysis along the lines of \cite{iyer:wald:1994PhRvD..50..846I}, see below, 
yields conserved charges corresponding to the Killing fields
$\partial_t, \partial_\phi$, which in turn correspond to the gauge invariant
deformations $\dot \massADM, \dot \angmom$ mentioned above.

The infinitesimal boosts, translations and
(non-axial) rotations of the black hole yield further non-radiating
modes which are, however, ``pure gauge'' in the sense that they are generated
by infinitesimal coordinate changes. 
If one imposes suitable regularity\footnote{The Kerr family of
line elements may be viewed as part of the type D family of vacuum metrics
which includes, among others, the NUT and C-metrics. See section \ref{sec:curved} for further discussion. The perturbations corresponding
e.g. to infinitesimal deformations of the NUT parameter are singular and may
thus be exluded by suitable regularity and decay conditions. See
\cite{virmani:2011PhRvD..84f4034V}, \cite{jezierski:1995GReGr..27..821J} for
remarks.} conditions on the perturbations
which exclude e.g. those which turn on the NUT charge, a 10-dimensional
space of non-radiating modes remains. This is spanned by
the 2-dimensional space of non-gauge modes which carry the $\dot \massADM,
\dot \angmom$ charges, together with the 
``pure gauge'' non-radiating modes, and 
corresponds in a natural way to the 
Lie algebra of the Poincare group.
It can be seen from this discussion that 
a combination of charge vanishing conditions and gauge conditions
allows one to eliminate all non-radiating solutions of linearized gravity. 

The constraint equations implied by the Maxwell and linearized gravity
equations are underdetermined elliptic systems, and therefore admit
solutions of compact support, see \cite{delay:2010arXiv1003.0535D} and
references therein. In particular, one may find solutions of the constraint
equations with arbitrarily rapid fall-off at infinity. The corresponding
solutions of the Maxwell equations have vanishing charges. For the case of
linearized gravity, the charges corresponding to $\dot \massADM, \dot
\angmom$ vanish for solutions of the field equations with rapid fall-off
at infinity. For such solutions, all non-radiating modes may
therefore be eliminated by imposing suitable gauge conditions. 

The following discussion may
easily be extended to the Einstein-Maxwell equations. 
Given an asymptotically flat vacuum spacetime
$(N, g_{ab})$,
a solution of the linearized Einstein equations $\dot g_{ab}$ (satisfying suitable asymptotic conditions) and a 
Killing field $\xi^a \partial_a$ we have that 
the variation of the Hamiltonian current is an exact form, which yields 
the relation 
\begin{equation}\label{eq:xicharge} 
\dot \PP_{\xi; \infty} = \int_S \dot \bQ[\xi] - \xi \cdot \bTheta \, .
\end{equation} 
Here, $\PP_{\xi; \infty}$ is the Hamiltonian charge at infinity, 
generating the action of $\xi$, $\bQ[\xi]$ is the Noether charge two-form for
$\xi$, and $\bTheta$ is the symplectic current three-form, defined with respect
to the variation $\dot g_{ab}$. We use a $\dot{\ }$ to denote variations
along $\dot g_{ab}$, thus $\dot \PP_{\xi;\infty}$ and $\dot \bQ[\xi]$ denote the
variation of the Hamiltonian and the Noether two-form, respectively. The
integral on the right hand side of \eqref{eq:xicharge} is evaluated over an
arbitrary sphere, which generates the second homology class.

For the case of $\xi = \partial_t$, and considering solutions of the
linearized Einstein equations on the Kerr background we have, following the
discussion above, 
$$
\dot \massADM = \dot \PP_{\partial_t; \infty} 
$$
Working with the Carter tetrad, let $\Psi_i$, $i=0,\cdots, 4$ be the Weyl
scalars and let $Z^I$, $I = 0,1,2$ denote the corresponding basis 
for the space of complex, self-dual two-forms, see section
\ref{sec:form} for details. 
In this paper we shall show that the natural linearization of the
spin-lowered Weyl tensor $\MM$ is the two-form
$$
 \linmass = \psi\dot\Psi_1 Z^0 + \psi\dot\Psi_2 Z^1 + \psi\dot\Psi_3 Z^2 +
 \tfrac{3}{2} \psi\Psi_2 \dot Z^1 .
$$
As will be demonstrated, see section \ref{sec:fack} below, 
$\linmass$ 
is closed, and hence the integral 
\begin{equation}\label{eq:linmasscharge} 
\int_S \linmass
\end{equation}
defines a conserved charge. 
A charge vanishing condition for the linearized mass, 
analogous to the
one discussed above for the charges of the Maxwell field, may be introduced
by requiring that this integral vanishes. The coordinate form of this 
charge vanishing condition is 
\begin{align} \label{eq:integrability2}
\int_{S^2(t,r)} \big( 2 V_L^{-1/2} \dot{\widehat \Psi}_2 + \i a \sin \theta
\dot \Psi_{diff} \big) (r-\i a \cos\theta) \d\mu = 0 , 
\end{align}
which should be compared to the corresponding condition for the 
Maxwell case, cf.  \eqref{eq:integrability}. 
Here, $\dot{\widehat \Psi}_2$ and $\dot \Psi_{diff}$ are suitable
combinations of the linearized curvature scalars $\dot \Psi_1,
\dot\Psi_2,\dot\Psi_3$ and linearized tetrad.

Let $\dot g_{ab}$ be a solution of the linearized Einstein equation on the
Kerr background, satisfying 
suitable asymptotic conditions, and let $\dot \massADM$ be the corresponding
perturbation of the ADM mass. Letting $S = S^2(t,r)$ and evaluating  the limit
of \eqref{eq:linmasscharge} as $r\to \infty$ one finds, in view of the fact
that \eqref{eq:linmasscharge} is conserved, the identity 
$$
\dot \massADM = \frac{1}{4\pi i} \int_S \linmass
$$
for any smooth 2-sphere $S$ in the exterior of the Kerr black hole. 
Thus we have the relation 
\begin{equation}\label{eq:twocharges} 
\int_S \dot \bQ[\partial_t] - \partial_t \cdot \bTheta 
= \frac{1}{4\pi i} \int_S \linmass
\end{equation}
for any surface $S$ in the Kerr exterior. We remark that the left hand side
of \eqref{eq:twocharges} can be evaluated in terms of the metric perturbation
using the expressions for $\bQ$ and $\bTheta$ given in \cite[section
  V]{iyer:wald:1994PhRvD..50..846I}. On the other hand, the right hand side
has been calculated in terms of linearized
curvature. It would be of interest to have a direct derivation of the
resulting identity. 

The canonical analysis following \cite{iyer:wald:1994PhRvD..50..846I}
which has been discussed above shows that in addition to the conserved charge
corresponding to $\dot \massADM$, equation \eqref{eq:xicharge} with $\xi =
\partial_\phi$, the angular Killing field,  
gives a conserved charge integral 
for linearized angular momentum $\dot \angmom$. 
If $\partial_\phi$ is tangent
  to $S$, then the term $\partial_\phi \cdot \bTheta$ does not contribute in
  \eqref{eq:xicharge}. 
We remark that an expression for $\dot \angmom$ for linearized gravity on the
Schwarzschild background was given in \cite[section 3]{jezierski:1999}.  
A charge integral for
$\dot \angmom$ for linearized gravity on the Kerr background will be 
considered in a future paper. 

\begin{remark} 
\begin{enumerate} 
\item There are many candidates for a quasi-local mass expression in the
literature including, to mention just a few, 
those put forward by Penrose, Brown and York, and Wang and Yau.
See the review of Szabados 
\cite{szabados:2004LRR.....7....4S} for background and references. Although
as discussed above, cf. equation \eqref{eq:mass-charge}, 
for a spacetime of type D, there is a quasi-local mass charge, 
it must be emphasized that 
for a general spacetime on cannot expect the existence of 
a quasi-local mass
which is \emph{conserved}, i.e. independent of the 2-surface used in
its definition. The same is true for linearized gravity
on a general background. Thus the existence of a conserved charge integral
for the linearized mass is a feature which is special to linearized gravity on a
background with Killing symmetries. 

\item If we consider linearized
gravity without sources, on the Minkowski background, 
the linearized mass must vanish due to
the fact that Minkowski space is topologically trivial. This reflects the
fact that when viewed as a function on the space of Cauchy data, the ADM mass
vanishes quadratically at the trivial data,
cf. \cite{ChB:fischer:marsden:1979igsg.conf..396C}. 
On the other hand, by the positive mass theorem, 
for any non-flat spacetime, asymptotic to Minkowski space in a suitable
sense, the ADM mass defined at infinity must be positive. 
\end{enumerate}
\end{remark}

This paper is organized as follows. In section \ref{sec:form}, we introduce
bivector formalism. Conformal Killing Yano tensors and Killing spinors are
discussed in section \ref{sec:cky}. Section \ref{sec:conservedcharge} deals with
conserved charges for spin-2 fields on Minkowski (\S \ref{sec:flat} ) and
type D spacetimes (\S \ref{sec:curved}). The main result, a charge integral
in terms of linearized curvature, is derived in section \ref{sec:fack}, and
finally, section \ref{sec:conclusions} contains some concluding remarks. 

\section{Preliminaries and notation} \label{sec:form}
Let $(N,g_{ab})$ be a 4 dimensional Lorentzian spacetime of
signature $+---$, 
admitting a spinor structure. 
Although most of the results can be generalized 
to the electrovac case with cosmological constant, we restrict in this paper
to the vacuum case. In particular, we consider test Maxwell fields and
linearized gravity on vacuum type D background spacetimes.

Let 
$o_A, \io_A$ be a spinor dyad, normalized so that $o_A \io^A = 1$, and let 
\begin{align*}
 \NPl^a = o^A \bar o^{A'}, && \NPm^a = o^A \bar \io^{A'}, && \NPmbar^a = \io^A \bar o^{A'}, && \NPn^a = \io^A \bar \io^{A'} 
\end{align*}
be the corresponding null tetrad, 
satisfying $\NPl^a \NPn_a = - \NPm^a \NPmbar_a = 1$, the other inner products
being zero. The 2-spinor calculus provides a powerful tool for computations
in 4-dimensional geometry. The GHP formalism deals with dyad (or
equivalently tetrad) components of geometric objects 
and exploits the simplifications arising by
taking into account the action of dyad rescalings and permutations. 
These formalisms are closely related to the less widely used 
\emph{bivector formalism} \cite{jordan:ehlers:sachs:1961:II,Bichteler:1964, Cahen:Debever:Defrise:1967,israel:bivector:book}
in which the basic quantity is a basis for the 3-dimensional space of complex
self-dual two-forms. A two-form $Z$ is called self-dual, if $*Z = \i Z$ and anti self-dual, if $*Z = -\i Z$. Given a spinor dyad, a
natural choice\footnote{We use the convention of \cite{fayos:ferrando:jaen:1990}, which differs from \cite{israel:bivector:book,fackerell:1982} by a factor of 2 in the middle component and the numbering.} 
is 
\begin{subequations}\begin{align} 
 Z^0_{ab} &= 2 \NPmbar_{[a} \NPn_{b]} = \io_\sa \io_\sb \ba\eps_{\sa'\sb'} \\
 Z^1_{ab} &= 2\NPn_{[a} \NPl_{b]} - 2\NPmbar_{[a} \NPm_{b]}  = -2 o_{(\sa} \io_{\sb)} \ba\eps_{\sa'\sb'} \\
 Z^2_{ab} &= 2 \NPl_{[a} \NPm_{b]}  = o_\sa o_\sb \ba\eps_{\sa'\sb'} \, ,
\end{align}\end{subequations}
where the notation $2x_{[a} y_{b]} = x_a y_b - y_a x_b$ for anti symmetrization and $2x_{(a} y_{b)} = x_a y_b + y_a x_b$ for symmetrization is used.
We use capital latin indices $I,J,K$ taking values in $0,1,2$ for the
elements in the bivector triad $Z^I$. 
The metric $g_{ab}$ induces a triad metric $G_{IJ}$ and its inverse
$G^{IJ}$ given by 
\begin{align*}
G^{IJ} = Z^I \cdot Z^J = \begin{pmatrix} 0 & 0 & 1 \\ 0 & -2 & 0 \\ 1 & 0 & 0 \end{pmatrix} , &&
G_{IJ} = \begin{pmatrix} 0 & 0 & 1 \\ 0 & -\thalf & 0 \\ 1 & 0 & 0 \end{pmatrix} \, .
\end{align*}
Here, $\cdot$ is the induced inner product on two-form, $ Z^I \cdot Z^J = \half Z^I{}_{ab}  Z^{Jab}$. Triad indices are raised and lowered with this metric,
\begin{align*}
 Z_{0} = Z^2 , && Z_{1} = -\thalf Z^1 , && Z_{2} = Z^0 .
\end{align*}
More general we have
\begin{prop} \label{prop:zz}
\begin{subequations}\begin{align}
 Z^J{}_a{}^c Z^K{}_{bc} &= \frac{1}{2} G^{JK}g_{ab} + \epsilon^{JKL} Z_{L ab}\\
 Z^J{}_{[a}{}^c \bar Z^K{}_{b]c} &= 0 \\
 Z^{Jab} \bar Z^K{}_{ab} &= 0
\end{align}\end{subequations}
with $\epsilon^{JKL}$ the totally antisymmetric symbol fixed by $\epsilon^{012}=1$.
\end{prop}

A real two-form $F_{ab}$, e.g. the Maxwell field strength, has spinor
representation
\begin{align*}
 F_{ab} = \phi_{\sa\sb} \eps_{\sa'\sb'} + \ba \phi_{\sa'\sb'} \eps_{\sa\sb}.
\end{align*}
It is equivalent to the symmetric 2-spinor $\phi_{\sa\sb} = \phi_2 o_\sa
o_\sb -2 \phi_1 o_{(\sa} \io_{\sb)} + \phi_0 \io_\sa \io_\sb$, where the six
real degress of freedom of $F_{ab}$ are encoded in 3 complex scalars 
\begin{align*}
 \phi_0 &= \phi_{\sa\sb} o^\sa o^\sb  = F_{ab}  \NPl^a \NPm^b = F \cdot Z_0\\
 \phi_1 &= \phi_{\sa\sb} \io^\sa o^\sb = \thalf F_{ab} ( \NPl^a \NPn^b - \NPm^a \NPmbar^b ) = F \cdot Z_1\\
 \phi_2 &= \phi_{\sa\sb} \io^\sa \io^\sb = F_{ab}  \NPmbar^a \NPn^b = F \cdot Z_2 \, .
\end{align*}
So the real two-form has bivector representation
\begin{align*}
F = \phi_0 Z^0 + \phi_1 Z^1 + \phi_2 Z^2 + \ba \phi_0 \ba Z^0 + \ba \phi_1 \ba Z^1 + \ba \phi_2 \ba Z^2,
\end{align*}
or in index notation $ \phi_I = F \cdot Z_I $ and $ F = \phi_I Z^I + \ba \phi_I \ba Z^I $.

The Weyl tensor is a symmetric 2-tensor over bivector space and has spinor representation
\begin{align*}
 -C_{abcd} = \Psi_{\sa\sb\sc\sd} \ba\eps_{\sa'\sb'} \ba\eps_{\sc'\sd'} + \ba\Psi_{\sa'\sb'\sc'\sd'} \eps_{\sa\sb} \eps_{\sc\sd} \, ,
\end{align*}
where $\Psi_{\sa\sb\sc\sd}$ is a completely symmetric 4-spinor. The 10 degrees of freedom of the Weyl tensor are given by 5 complex scalars\footnote{Due to its symmetries, the Weyl tensor is a symmetric two-tensor over the space of two-forms. The induced inner product is $C \cdot (Z_I, Z_J) = \frac{1}{4} C_{abcd} Z_I^{ab} Z_J^{cd}$.}
\begin{alignat*}{3}
 \Psi_0 &= \Psi_{\sa\sb\sc\sd} \, o^\sa o^\sb o^\sc o^\sd &&= -C_{abcd} \NPl^a \NPm^b \NPl^c \NPm^d &&= -C \cdot (Z_0,Z_0)\\
 \Psi_1 &= \Psi_{\sa\sb\sc\sd} \, o^\sa o^\sb o^\sc \io^\sd &&= -C_{abcd} \NPl^a \NPn^b \NPl^c \NPm^d &&= -C \cdot (Z_0,Z_1)\\ 
 \Psi_2 &= \Psi_{\sa\sb\sc\sd} \, o^\sa o^\sb \io^\sc \io^\sd &&= -C_{abcd} \NPl^a \NPm^b \NPmbar^c \NPn^d &&= -C \cdot (Z_0,Z_2)= -C \cdot (Z_1,Z_1)\\ 
 \Psi_3 &= \Psi_{\sa\sb\sc\sd} \, o^\sa \io^\sb \io^\sc \io^\sd &&= -C_{abcd} \NPl^a \NPn^b \NPmbar^c \NPn^d &&= -C \cdot (Z_2,Z_1)\\
 \Psi_4 &= \Psi_{\sa\sb\sc\sd} \, \io^\sa \io^\sb \io^\sc \io^\sd &&= -C_{abcd} \NPn^a \NPmbar^b \NPn^c \NPmbar^d &&= -C \cdot (Z_2,Z_2) \, .
\end{alignat*}
Similarly we could have used the Weyl 2-bivector 
\begin{align*}
C_{IJ} = - \frac{1}{4} C_{abcd} Z_I^{ab} Z_J^{cd} =  
\begin{pmatrix} 
\Psi_0 & \Psi_1 & \Psi_2 \\  
\Psi_1 & \Psi_2 & \Psi_3 \\
\Psi_2 & \Psi_3 & \Psi_4
\end{pmatrix} 
\end{align*}
which relates to the real Weyl tensor via
\begin{align} \label{eq:curv}
-C_{abcd} = C_{IJ} Z^I_{ab} \otimes Z^J_{cd} + \ba C_{IJ} \ba Z^I_{ab} \otimes \ba Z^J_{cd} \, .
\end{align}

Because of different conventions and normalisations in the literature \cite{jordan:ehlers:sachs:1961:II,Bichteler:1964, Cahen:Debever:Defrise:1967,israel:bivector:book}, we rederive here the equations of structure in bivector formalism. Based on Cartan's equations of structure for tetrad one-forms 
\footnote{Connection and curvature are defined by $\omega^a{}_{b\mu}= e^a{}_\nu \nabla_\mu e_b{}^\nu$ and $\Omega^a{}_{b\mu\nu} = 2 e^a{}_\sigma \nabla_{[\mu}\nabla_{\nu]} e_b{}^\sigma$, respectively.}
\begin{align} \label{eq:cartan1}
 \d e^a = - \omega^a{}_b \wedge e^b &&& 
 \Omega^a{}_b = \d \omega^a{}_b + \omega^a{}_c \wedge \omega^c{}_b \, ,
\end{align}
Bianchi identities
\begin{align} \label{eq:bianchi1}
 \Omega^a{}_b \wedge e^b = 0  &&&  
 \d \Omega^a{}_b = \Omega^a{}_c \wedge \omega^c{}_b - \omega^a{}_c \wedge \Omega^c{}_b \, ,
\end{align}
and definitions of connection one-forms $\sigma_J$ and curvature two-forms $\Sigma_J$  in bivector formalism,
\begin{align} \label{eq:conncurv}
 \omega_{ab} \, e^a \wedge e^b = -2 \sigma_J Z^J - 2 \bar \sigma_J \bar Z^J &&& 
 \Omega_{ab} e^a \wedge e^b = - 2 \Sigma_J Z^J - 2 \bar\Sigma_J \bar Z^J, 
\end{align} 
we find
\begin{prop} \label{prop:biveq}
The bivector equations of structure are
\begin{align} \label{eq:cartan2}
 \d Z^J = -2 \epsilon^{J K L} \sigma_K \wedge Z_L &&& 
\Sigma_J = \d \sigma_J + \half \epsilon_{JKL} \sigma^K \wedge \sigma^L
\end{align}
while the Bianchi identities read
\begin{align} \label{eq:bianchi2}
 \Sigma_{[J} \wedge Z_{K]} = 0 &&& 
 d\Sigma_J = - \epsilon_{JKL} \Sigma^K \wedge \sigma^L \, .
\end{align} 
Here $\wedge$ is the usual wedge product of one-forms $\sigma^J$ and two-forms $Z^J, \Sigma^J$.
\end{prop}
\begin{proof}
Expanding the bivectors $Z^J = \half Z^J _{ab} e^a \wedge e^b$, we find
\begin{align*}
 \d Z^J 
&= \half Z^J _{ab}  \left( \d e^a \wedge e^b - e^a \wedge \d e^b \right) 
= Z^J_{a b} \d e^a \wedge e^b \\
&= - Z^J_{a b} \, \omega^a{}_c e^c \wedge e^b \\
&= Z^J_{a b} \left( \sigma_K Z^{Ka}{}_c + \bar \sigma_K \bar Z^{Ka}{}_c \right) \wedge e^c \wedge e^b \\
&= \epsilon^{JKL} Z_{Lbc} \sigma_K \wedge e^c \wedge e^b \\
&= -2 \epsilon^{JKL} \sigma_K \wedge Z_L \, 
\end{align*}
where proposition \ref{prop:zz} has been used in the third step.
For the second equation of structure, we plug \eqref{eq:conncurv} into \eqref{eq:cartan1},
\begin{align*}
 -\Sigma_J Z^J_{ab} - \bar\Sigma_J \bar Z^J_{ab} 
&= -\d\sigma_J Z^J_{ab} - \d \bar\sigma_J \bar Z^J_{ab} + (\sigma_J Z^J_{ac} + \bar\sigma_J \bar Z^J_{ac}) \wedge (\sigma_K Z^{Kc}{}_{b} + \bar\sigma_K \bar Z^{Kc}{}_{b}) \, .
\end{align*}
Since $Z^J \cdot \bar Z^K = 0$ and proposition \ref{prop:zz}, the selfdual part reads
\begin{align*}
 \Sigma_J Z^J_{ab} = \d \sigma_J Z^J_{ab} + \epsilon^{KLJ} Z_{J ab} \sigma_K \wedge \sigma_L \, .
\end{align*}
Changing index positions by using $\det G_{JK} = \thalf$ gives the 2nd equation of structure. For the first Bianchi identity, look at
\begin{align*}
 0 &= \d^2 Z^J \\
 &= -2 \epsilon^{JKL} \left( \d\sigma_K \wedge Z_L - \sigma_K \wedge \d Z_L \right) \\
 &= -2 \epsilon^{JKL} \left( \Sigma_K \wedge Z_L - \half \epsilon_{KNM} \sigma^N \wedge \sigma^M \wedge Z_L + \sigma_K \wedge \epsilon_{LNM} \sigma^N \wedge Z^M \right) \\
 &= -2 \epsilon^{JKL} \Sigma_K \wedge Z_L \underbrace{+ \sigma^L \wedge \sigma^J \wedge Z_L - \sigma^J \wedge \sigma^L \wedge Z_L -2 \sigma_L \wedge \sigma^J \wedge Z^L}_{=0} + 2 \underbrace{\sigma_K \wedge  \sigma^K}_{=0} \wedge Z^J
\end{align*}
where the identity $\epsilon^{IJK} \epsilon_{INM} = \delta^J_N \delta^K_M - \delta^J_M \delta^K_N$ has been used. Finally, the second Bianchi identity is
\begin{align*}
 \d \Sigma_J 
 &= -\epsilon_{JKL} \d \sigma^K \wedge \sigma^L \\
 &= -\epsilon_{JKL} (\Sigma^K - \epsilon^{KMN} \sigma_M \wedge \sigma_N) \wedge \sigma^L \\
 &= -\epsilon_{JKL} \Sigma^K \wedge \sigma^L + \underbrace{\sigma_L \wedge \sigma_J \wedge \sigma^L}_{=0} -  \sigma_J \wedge \underbrace{\sigma_L \wedge \sigma^L}_{=0} \, .
\end{align*}
\end{proof}
\begin{remark}
 Instead of using Cartan equations for the tetrad one could have used the bivector connection form
\begin{align}
\omega_{IJa} := \eps_{IJK} \sigma^K_a =  Z_{[J}^{bc} \nabla_a Z_{I]bc} \, .
\end{align}
\end{remark}

For later use it is convenient to write the components of the equations of structure explicitely. The connection one-forms for example can be expressed in terms of NP spin
coefficients,
\begin{subequations}\begin{alignat}{2}
\sigma_{0a} &= \NPm^b \nabla_a \NPl_b &&= \tau \NPl_a + \kappa \NPn_a - \rho \NPm_a - \sigma \NPmbar_a \\
\sigma_{1a} &= \half \left( \NPn^b \nabla_a \NPl_b - \NPmbar^b \nabla_a \NPm_b \right)  &&=-\epsilon' \NPl_a + \epsilon \NPn_a + \beta' \NPm_a - \beta \NPmbar_a \\
\sigma_{2a} &= -\NPmbar^b \nabla_a \NPn_b &&= -\kappa' \NPl_a - \tau' \NPn_a + \sigma' \NPm_a + \rho' \NPmbar_a \, .
\end{alignat}\end{subequations}
The middle component $\sigma_{1a}$ collects all unweighted coefficients and so can be used to define the GHP covariant derivative $\Theta_a \eta = (\nabla_a - p \sigma_{1a} - q \ba \sigma_{1a}) \eta$. To avoid clutter in the notation, we write $\Gamma:=\sigma_0$ and $\sigma_2 = -\Gamma'$, where $'$ is the GHP prime operation\cite{GHP}. Derivatives of the spinor dyad can now be written in the compact form $\Theta_a o^A = - \Gamma_a \io^A$ and $\Theta_a \io^A = - \Gamma'_a o^A$, and the components of the first equations of structure, which we present here for convenience with the usual exterior derivative and with weighted exterior derivative $\d^\Theta = \d - p \sigma_1 \wedge - q \ba \sigma_1 \wedge$, read
\begin{subequations}\begin{align} 
\d^\Theta Z^0 &= \Gamma' \wedge Z^1  & \Leftrightarrow & & \d Z^0 &= -2 \sigma_1 \wedge  Z^0 + \Gamma' \wedge Z^1\\
\d^\Theta Z^1 &= 2 \Gamma \wedge Z^0 + 2 \Gamma' \wedge Z^2 & \Leftrightarrow & & \d Z^1 &= 2 \Gamma \wedge Z^0 + 2 \Gamma' \wedge Z^2\\
\d^\Theta Z^2 &= \Gamma \wedge Z^1 & \Leftrightarrow & & \d Z^2 &= 2 \sigma_1 \wedge Z^2 + \Gamma \wedge Z^1 .
\end{align}\end{subequations}
Note that the middle component can be simplified to $\d Z^1 = - h \wedge Z^1 $ with the one-form $ h =  2 (\rho' \NPl + \rho \NPn -\tau' \NPm - \tau \NPmbar)$. This fact and a relation between type D curvature $\Psi_2$ and $h$ will be crucial in the derivation of the conservation law in section \ref{sec:fack}.

In vacuum, we have for the curvature two-forms $\Sigma_J= C_{JK} Z^K$ and the components of the second equations of structure read
\begin{subequations}\begin{align}
 \Sigma_0 &= C_{0J}Z^{J} = \d^\Theta \Gamma = \d \Gamma - 2 \sigma_1 \wedge \Gamma \\
 \Sigma_1 &= C_{1J}Z^{J} = \d \sigma_1 - \Gamma \wedge \Gamma' \\
 \Sigma_2 &= C_{2J}Z^{J} = -\d^\Theta \Gamma' = -\d \Gamma' - 2 \sigma_1 \wedge \Gamma' \, .
\end{align}\end{subequations}
Finally the Bianchi identities are
\begin{subequations}\begin{align}
\d^\Theta \Sigma_0 &= -2\Gamma \wedge \Sigma_1 
& \Leftrightarrow & &\d \Sigma_0 &= 2 \sigma_1 \wedge \Sigma_0 - 2 \Gamma \wedge \Sigma_1 \\
\d^\Theta \Sigma_1 &= - \Gamma' \wedge \Sigma_0 - \Gamma \wedge \Sigma_2  
& \Leftrightarrow & &\d \Sigma_1 &= - \Gamma' \wedge \Sigma_0 - \Gamma \wedge \Sigma_2 \label{eq:bianchi}\\
\d^\Theta \Sigma_2 &= - 2 \Gamma' \wedge \Sigma_1 
& \Leftrightarrow & & \d \Sigma_2 &= -2 \sigma_1 \wedge \Sigma_2 - 2 \Gamma' \wedge \Sigma_1 . 
\end{align}\end{subequations}

\section{Conformal Killing Yano tensors and Killing spinors} \label{sec:cky}

Conformal Killing Yano tensors of rank 2 are two-forms $Y_{ab}$ solving the conformal Killing Yano equation,
 \begin{align} \label{eq:cyt}
 Y_{a(b;c)} = g_{bc}\xi_a -g_{a(b}\xi_{c)}, \text{ where } \xi_a = \tfrac{1}{3} Y_a{}^b{}_{;b} .
\end{align}
It is well known, that the divergence $\xi^a$ is a Killing vector and in
case it vanishes, $Y_{ab}$ is called Killing Yano tensor. The symmetrised
product $X_{c(a}Y_{b)}{}^c =:K_{ab}$ of Killing Yano tensors
$X_{ab}, Y_{ab}$ is a Killing tensor, $\nabla_{(a} K_{bc)} = 0$,
which can be used to construct a constant of motion or a symmetry operator for e.g. the scalar
wave equation, known as Carter's constant and Carter operator, respectively. By inserting $Y_{ab} = \kappa_{AB} \eps_{A'B'} + \bar
\kappa_{A'B'} \eps_{AB}$ into \eqref{eq:cyt} one can show that $\kappa_{AB}$
and $\bar \kappa_{A'B'}$ satisfy the Killing spinor equation 
\begin{align} \label{eq:ks}
 \nabla_{A'(A} \kappa_{BC)} = 0
\end{align}
and its complex conjugated version. For the spinor components $\kappa_{AB} = \kappa_2
o_A o_B -2 \kappa_1 o_{(A} \io_{B)} + \kappa_0 \io_A \io_B  $ (or
equivalently the self dual bivector components of $Y_{ab}$, we find the
following set of eight scalar equations 
\begin{align}
 \begin{aligned} \label{eq:kscomp1}
 \tho \kappa_0 = -2 \kappa \kappa_1 , &&& \edt \kappa_0 = -2 \sigma \kappa_1  ,&&& 
 {\tho}' \kappa_2 = -2 \kappa' \kappa_1 , &&& {\edt}' \kappa_2 = -2 \sigma' \kappa_1 
\end{aligned}\\
\begin{aligned} \label{eq:kscomp2}
({\edt}' + 2\tau')\kappa_0 +2 (\tho + \rho)\kappa_1 = -2 \kappa \kappa_2 , &&&
({\tho}' + 2\rho')\kappa_0 +2 (\edt + \tau) \kappa_1 = -2 \sigma \kappa_0 \\
({\edt} + 2\tau)\kappa_2 +2 ({\tho}' + \rho')\kappa_1 = -2 \kappa' \kappa_0 , &&&
({\tho} + 2\rho)\kappa_2 +2 ({\edt}' + \tau') \kappa_1 = -2 \sigma' \kappa_2 \,,
\end{aligned} 
\end{align}
by projecting \eqref{eq:ks} into a spinor dyad. Thus, we have three different sets of equations, \eqref{eq:cyt}, \eqref{eq:ks}, 
(\ref{eq:kscomp1},\ref{eq:kscomp2}), which are equivalent and we will use the most appropriate for the problem at hand.

As spin-s fields are heavily restricted on curved backgrounds (Buchdahl constraint, see equation (5.8.2) in \cite{PR:I}), so are Killing spinors. Consider a Killing spinor $\kappa_{A_1 ... A_n} = \kappa_{(A_1 ... A_n)}$ which satisfies the  Killing spinor equation of valence $n$
\begin{align}
 \nabla_{B'(B}\kappa_{A_1 ... A_n)} = 0 \, .
\end{align}
Contracting a second derivative $\nabla^{B'}{}_C$ and symmetrising gives
\begin{align*}
 0 &= \nabla^{B'}{}_{(C} \nabla_{|B'|B} \kappa_{A_1 ... A_n)} \\
 &= -\Box_{(BC} \kappa_{A_1 ... A_n)} \\
 &= \Psi_{(BCA_1}{}^D \kappa_{D A_2 ... A_n)} + \dots + \Psi_{(BCA_n}{}^D \kappa_{A_1 ... A_{n-1}D)} \\
 &= n\Psi_{(BCA_1}{}^D \kappa_{D A_2 ... A_n)} \, .
\end{align*}
For Killing spinors of valence 1 (satisfying the twistor equation) this yields $ 0 = \Psi_{ABCD} \kappa^D$ as can be found in \cite{PR:II}, eq.(6.1.6). For 2-spinors we find
\begin{align} \label{eq:IntCond}
 0 = \Psi_{(ABC}{}^D \kappa_{DE)} \, .
\end{align}
For non trivial $\kappa$, this restricts the spacetime to be of Petrov type $D,N$ or $O$.
For a given spacetime of type D in a principal frame (only $\Psi_2 \neq 0$) \eqref{eq:IntCond} becomes
\begin{align*}
 0 &= \Psi_2 \, o_{(A} o_B \io_C \io_D \left( \kappa_0 \io^D \io_{E)} + \kappa_1 o^D \io_{E)} + \kappa_1 \io^D o_{E)} + \kappa_2 o^D o_{E)} \right) \\
 &= \Psi_2 \left( C_1 \kappa_0 o_{(A} \io_B \io_C \io_{E)} + C_2 \kappa_2 \io_{(A} o_B o_C o_{E)} \right)
\end{align*}
with constants $C_1,C_2$ and it follows $\kappa_0 \equiv 0 \equiv
\kappa_2$. The remaining component satisfies the simplified equations
\eqref{eq:typeDKS}, which have only one non trivial complex solution, cf. \cite{Glass:1996} where explicit integration of the conformal Killing Yano equation was done.

\section{Conserved Charges} \label{sec:conservedcharge} 

\subsection{Conserved charges for Minkowski spacetime} \label{sec:flat}

The Killing spinor equation or conformal Killing Yano equation on Minkowski space has been widely discussed in the literature \cite{PR:II},\cite{Jezierski:Lukasik:2006}, \cite{herdegen:1991} and the explicit solution in cartesian coordinates is well known,
\begin{align} \label{eq:kscart}
 \kappa^{AB} = U^{AB} + 2 x^{A'(A} V^{B)}_{A'} + x^{A'A}x^{B'B} W_{A'B'} .
\end{align}
Here  $U^{AB},W_{A'B'}$ are constant, symmetric spinors and $V_{A'}^B$ a constant complex vector which yield $2 \cdot 6 + 8 = 20$ independent real solutions. Each solution gives a charge when contracted into a spin-2 field, e.g. the linearized Weyl tensor, and integrated over a 2-sphere. In \cite[p.99]{PR:II}, 10 of these charges are related to a source for linearized gravity in the following sense. Given a divergence free, symmetric energy momentum tensor $T_{ab}$, one has for each Killing field $\xi^b$ the divergence free current $J_a = T_{ab} \xi^b$. Using linearized Einstein equations
\begin{align}
 \dot G_{ab} = \dot R_{acb}{}^c - \frac{1}{2}g_{ab} \dot R_{cd}{}^{cd} = -8 \pi G \dot T_{ab}
\end{align}
and the conformal Killing Yano equation \eqref{eq:cyt}, they showed
\begin{align}
 3 \int_{\partial \Sigma} \dot R_{abcd} *\hspace{-4pt}Y^{cd} \d x^a \wedge \d x^b = 16 \pi G \int_\Sigma e_{abc}{}^d \dot T_{df} \xi^f \d x^a \wedge \d x^b \wedge \d x^c .
\end{align}
Here $\Sigma$ denotes a 3 dimensional hypersurface with boundary $\partial\Sigma$ and $e_{abcd}$ is the Levi-Civita tensor. The left hand side is the charge integral described above, while the right hand side gives the more familiar form of a conserved three-form corresponding to a linarized source and a Killing vector $\xi^a = \tfrac{1}{3} Y^{ab}{}_{;b}$. Note that it is the dual conformal Killing Yano tensor on the left hand side, which gives the charge associated to the isometry $\xi^a$. In cartesian coordinates  $x^a = (t,x,y,z)$ the Poincar\'{e} isometries read
\begin{align}
 \mathcal{T}_a = \frac{\partial}{\partial x^a}  && \mathcal{L}_{ab} = x_a \frac{\partial}{\partial x^b} - x_b \frac{\partial}{\partial x^a}
\end{align}
and the relation to the charges is listed in table \ref{tab:isometries}. The angular momentum around the $z$-axis is found in the component $\mathcal{L}_{xy} = \partial_\phi$.
\begin{table}
\caption{Poincar\'{e} isometries and corresponding charges}
\begin{tabular}{llll}
\toprule
label& isometry & charge& \# \\
\hline
$\mathcal{T}_t$ & time translation & mass & 1\\
$\mathcal{T}_i$ & spatial translations & linear momenta & 3\\
$\mathcal{L}_{ij}$ & rotations &  angular momenta & 3 \\
$\mathcal{L}_{ti}$ & boosts & center of mass & 3\\
\bottomrule
\end{tabular} \label{tab:isometries}
\end{table}
Explicit expressions for linearized sources generating these charges can be found in \cite[eq.27]{jezierski:1995GReGr..27..821J}.

The 10 remaining charges cannot be generated this way, since the
corresponding conformal Killing Yano tensors have vanishing divergence (they
are Killing Yano tensors). One of these charges corresponds to the NUT
parameter\footnote{sometimes called dual mass, because of duality rotation from Schwarschild to NUT, see the appendix of \cite{Ramaswamy:Sen:1981}}, 
and the remaining nine are three
dual linear momenta and six ofam\footnote{Obstructions for angular
  momentum, see \cite{jezierski:2002CQGra..19.4405J}.}.
In the expression \eqref{eq:kscart} for a general Killing spinor, 
they correspond to $U$ and the
imaginary part of $V$. For a metric perturbation, which one might interpret
as a potential for the linarized curvature, these 10 additional charges
vanish, see \cite[\S 6.5]{PR:II}. 

To understand the charges as projections into $l=0$ and $l=1$ mode, we rederive the complete set of solutions in spherical coordinates using spin weighted spherical harmonics. A null tetrad for Minkowski spacetime in spherical coordinates $(t,r,\theta,\phi)$ (symmetric Carter tetrad) is given by
\begin{align*}
 \NPl^a = \frac{1}{\sqrt{2}} \bigg[1,1,0,0 \bigg] , &&
\NPn^a = \frac{1}{\sqrt{2}}\bigg[1,-1,0,0 \bigg] , &&
\NPm^a = \frac{1}{\sqrt{2}r}\bigg[0,0,1,\frac{\i}{\sin\theta} \bigg] ,
\end{align*}
with non vanishing spin coefficients
\begin{align*}
 \rho = -\frac{1}{\sqrt{2}r} = -\rho' , && \beta = \frac{\cot\theta}{2\sqrt{2}r} = \beta' .
\end{align*}
A general two-form can be expanded
\begin{align*} 
 Y =&+ \kappa_2 \tfrac{r}{2}(\d r - \d t) \wedge (\d\theta + \i  \sin\theta \, \d\varphi )\\ &- \kappa_1 ( \d t \wedge \d r + \i r^2 \sin\theta \, \d \theta \wedge \d \varphi )\\  &+ \kappa_0 \tfrac{r}{2}(\d r + \d t) \wedge (\d\theta - \i\sin\theta \, \d\varphi ) + c.c.
\end{align*}
and it is a conformal Killing Yano tensor, if the components $\kappa_i$ satisfy (\ref{eq:kscomp1},\ref{eq:kscomp2}). The subset \eqref{eq:kscomp1} of the Killing spinor equation becomes
\begin{align*}
 \left(\partial_t + \partial_r \right) \kappa_0 = 0, && \left(\partial_\theta + \frac{\i}{\sin\theta}\partial_\varphi - \cot\theta \right) \kappa_0 = 0, \\
 \left(\partial_t - \partial_r \right) \kappa_2 = 0, && \left(\partial_\theta - \frac{\i}{\sin\theta} \partial_\varphi - \cot\theta \right)\kappa_2 = 0,
\end{align*}
so $\kappa_0 = f_0(t-r) \Y{1}{1}{m}$ and $\kappa_2 = f_1(t+r) \Y{-1}{1}{m}$ with functions $f_i$ depending on advanced and retarded coordinates only. Finally \eqref{eq:kscomp2} can be solved for $\kappa_1$, which is only possible for particular functions $f_i$. The result is given in table \ref{tab:MinkKS}.
\begin{table} 
\caption{Solutions to the Killing spinor equation on Minkowski spacetime in
  spherical coordinates.} \label{tab:KS}
 \centering
\begin{tabular}{lccc||ccc} 
\toprule
& \multicolumn{3}{c}{components} & & \multicolumn{2}{c}{divergence} \\
label & $\kappa_0 / \sqrt{2}$ & $\kappa_1$ & $\kappa_2 / \sqrt{2}$ & combination & $\Re$ & $\Im$ \\
\hline
\rowcolor{g1}
&&& & $\Omega^0_{-1}$ & 0 & 0 \\ \rowcolor{g1}
$\Omega^0_m$ & $\Y{1}{1}{m}$ & $\Y{0}{1}{m}$ & $\Y{-1}{1}{m}$ & $\Omega^0_{0}$ & 0 & 0 \\ \rowcolor{g1}
&&& & $\Omega^0_{1}$ & 0 & 0 \\ \rowcolor{g1}
\rowcolor{g2}
&&& & $\Omega^1$ & $\mathcal{T}_t$ & 0 \\ \rowcolor{g2}
$\Omega^1$ & 0 & r & 0 & $\Omega^1_1 - \Omega^1_{-1}$ & $\mathcal{T}_x$ & 0 \\ \rowcolor{g2}
$\Omega^1_m$ & $(t-r) \Y{1}{1}{m}$ & $t \Y{0}{1}{m}$ & $(t+r) \Y{-1}{1}{m}$ & $\i\Omega^1_1 + \i\Omega^1_{-1}$ & $\mathcal{T}_y$ & 0\\ \rowcolor{g2}
&&& & $\Omega^1_0$ & $\mathcal{T}_z$ & 0 \\ \rowcolor{g2}
\rowcolor{g3}
&&& & $\Omega^2_1 - \Omega^2_{-1}$ & $\mathcal{L}_{tx}$ & $\mathcal{L}_{yz}$ \\ \rowcolor{g3}
$\Omega^2_m$ & $(t-r)^2 \Y{1}{1}{m}$ & $(t^2-r^2) \Y{0}{1}{m}$ & $(t+r)^2\Y{-1}{1}{m}$ & $\i\Omega^2_1 + \i\Omega^2_{-1}$ &  $\mathcal{L}_{ty}$ & $\mathcal{L}_{xz}$ \\ \rowcolor{g3}
&&& & $\Omega^2_0$ & $\mathcal{L}_{tz}$ & $\mathcal{L}_{xy}$ \\
\bottomrule
\end{tabular} \label{tab:MinkKS}
\end{table}
$\Omega^1$ is one complex solution, while $\Omega^i_m, i=0,1,2$ represent 3 complex solutions each, ($m=0,\pm 1$). We find the following correspondence to the solutions \eqref{eq:kscart} in cartesian coordinates
\begin{align*}
 \Omega^0_m \leftrightarrow U^{AB}, && \Omega^1,\Omega^1_m \leftrightarrow V_{A'}^A, && \Omega^2_m \leftrightarrow W_{A'B'}.
\end{align*}

\subsection{Conserved charges for type D spacetimes} \label{sec:curved}

The vacuum field equations in the algebraically special case of Petrov type D
have been integrated explicitly by Kinnersley \cite{kinnersley:1969}. An
explicit type D line element solving the Einstein-Maxwell 
equations with cosmological
constant is known, from which all type D line elements of this type 
can be derived by certain limiting procedures,
see \cite[\S 19.1.2]{exact:solutions:book}, see also 
\cite{Debever:Kamran:McLenaghan:1984}. The family of
type D spacetimes contains the Kerr and Schwarzschild solutions, but also
solutions with more complicated topology and asymptotic behaviour, such 
as the NUT- or C-metrics, and solutions whose orbits of the isometry group are
null. In the following, we again restrict to the vacuum case.

A Newman-Penrose tetrad such that the 
two real null vectors $l^a, n^a$ are aligned with
the two repeated principal null directions of a Weyl tensor of Petrov type 
D is called a principal tetrad. In this case,
\begin{align*}
 \Psi_0 = \Psi_1 = 0 = \Psi_3 = \Psi_4, && \kappa = \kappa' = 0  = \sigma = \sigma'
\end{align*}
and $\Psi_2 \neq 0$. 
Due to the integrability condition \eqref{eq:IntCond}, 
we have $\kappa_0 = 0 = \kappa_2$. Hence, the
components (\ref{eq:kscomp1},\ref{eq:kscomp2}) of the Killing spinor equation
simplify  to 
\begin{align} \label{eq:typeDKS}
 (\tho + \rho)\kappa_1 = 0 , &&
 (\edt + \tau) \kappa_1 = 0, && 
 ({\tho}' + \rho')\kappa_1 = 0 , &&
 ({\edt}' + \tau') \kappa_1 = 0 .
\end{align}
Comparison with the Bianchi identities
\begin{align} \label{eq:typeDbianchi}
 (\tho - 3\rho)\Psi_2 = 0 , &&
 (\edt - 3\tau) \Psi_2 = 0, && 
 ({\tho}' - 3\rho')\Psi_2 = 0 , &&
 ({\edt}' - 3\tau') \Psi_2 = 0, 
\end{align}
shows that $\kappa_1 := \psi \propto \Psi_2^{-1/3}$ is a solution, and in fact
up to a constant $\kappa_{AB} = \psi o_{(A} \iota_{B)}$ is the only solution
of the Killing spinor equation. 

The divergence $\xi^{AA'} =
  \nabla^{A'}{}_B \kappa^{AB}$ is a Killing vector field, which is
  proportional to a real Killing vector field for all 
type D spacetimes except for Kinnersley class IIIB,
cf. \cite{collinson:1976}. If $\xi^{AA'}$ is real, 
the imaginary part of $\kappa_{AB}$ 
is a Killing-Yano tensor. Spacetimes satisfying the just
mentioned condition are called 
  \textit{generalized Kerr-NUT} spacetimes \cite{ferrando:saez:2007JMP....48j2504F}. 
The square of the Killing-Yano tensor is a symmetric Killing tensor $K_{ab} =
  Y_{ac} Y^c{}_b$ and it follows, that $\eta^a = K^{ab} \xi_b$ is a Killing
  vector. On a Kerr background, $\xi^a$ and $\eta^a$ are linearly independent
  and span the space of isometries, see \cite{Hughston:Sommers:1973}. In the
  special case of a Schwarzschild background, $\eta^a$ vanishes, see also
  \cite{Collinson:Smith:1977} for details. 

For Kerr spacetime in Boyer-Lindquist coordinates we find 
\begin{align*}
 \Psi_2 = - \frac{M}{(r - \i a \cos\theta)^3} ,  && \psi \propto r - \i a \cos\theta
\end{align*}
and we set the factor of proportionality to 1, so that the solution 
\begin{align} \label{eq:kssolution}
 \kappa_0 = 0 && \kappa_1 = \psi && \kappa_2 = 0
\end{align} 
reduces to $\Omega^1$ as given in table \ref{tab:KS},
in the Minkowski limit $M,a \to 0$.  We find
$ \nabla_b \left( \psi Z^{1ab} \right) = 3\left( \partial_t \right)^a $. 
The Killing spinor with components given by \eqref{eq:kssolution} is 
\begin{equation}\label{eq:kappa}
\kappa_{AB} = -2\psi o_{(A} \io_{B)} , 
\end{equation}
We have
$ \psi Z^1_{ab} = \kappa_{AB} \eps_{A'B'} $ and therefore
$$
\left( \partial_t \right)_a = \frac{1}{3} \nabla^b \left( \psi Z^{1}_{ab} \right) = - \frac{2}{3} \nabla^{B'B} ( \kappa_{AB} \eps_{A'B'} ) = \frac{2}{3} \nabla_{A'}{}^{B}  \kappa_{AB} \, .
$$
Spin lowering the Weyl spinor using \eqref{eq:kappa} gives the Maxwell
  field $\psi_{ABCD} \kappa^{CD}$, which has charges proportional to 
\textit{mass} and \textit{dual mass}, see also
\cite{jezierski:lukasik:2009}. 
Letting $\MM(C,\kappa)$ denote the 
corresponding closed complex two-form we have 
\begin{equation} \label{eq:maxwell}
 \MM(C,\kappa)= \psi \Psi_2 Z^1 .
\end{equation}
Evaluating the charge for the Kerr metric yields 
\begin{align}
 \frac{1}{4\pi\i} \int_{S^2} \MM(C,\kappa) = \frac{1}{4\pi\i} \int_{S^2} - \frac{M}{(r-\i a \cos\theta)^2} (-\i)(r^2+a^2) \sin\theta \d\theta \wedge \d\varphi = M,
\end{align}
where $M$ is the ADM mass while the dual mass is zero.

The closed two-form \eqref{eq:maxwell} has been derived much earlier by Jordan, Ehlers and Sachs \cite{jordan:ehlers:sachs:1961:II}. We will repeat the derivation here, since this formulation can be generalized to linearized gravity most easily. On a type D background, the curvature forms and the connection simplify to 
\begin{align}
 \Sigma_0 = \Psi_2 Z^2 && \Sigma_1 = \Psi_2 Z^1 && \Sigma_2 = \Psi_2 Z^0 && \Gamma = \tau l - \rho m \, ,
\end{align}
so the middle Bianchi identity \eqref{eq:bianchi} becomes
\begin{align*}
 2\d \Sigma_1 &= 2\Psi_2 \left[ (\rho' \bar m - \tau' n) \wedge l \wedge m + (\rho m - \tau l) \wedge \bar m \wedge n \right] \\
&=  2 \Psi_2 (\rho' \NPl + \rho \NPn -\tau' \NPm - \tau \NPmbar) \wedge Z^1 \\
&= h \wedge \Sigma_1 \, ,
\end{align*}
where $ h =  2 (\rho' \NPl + \rho \NPn -\tau' \NPm - \tau \NPmbar)$ was used. 
As noted in \cite{fayos:ferrando:jaen:1990}, the Bianchi identities \eqref{eq:typeDbianchi} can be rewritten as $2 \d \Psi_2 = 3 h \Psi_2$ and one obtains 
\begin{align*}
\d(\Psi_2 Z^1) =  \d\Sigma_1 = \half h \wedge \Sigma_1 = \frac{1}{3} \d\Psi_2 \wedge Z^1 \, .
\end{align*}
We finally end up with the \textit{Jordan-Ehlers-Sachs conservation law}\cite{jordan:ehlers:sachs:1961:II},
\begin{align} \label{eq:jes}
\d \left( \Psi_2^{2/3} Z^1 \right) = 0 \, .
\end{align}
Using $\psi \propto \Psi_2^{-1/3}$, this is the same result as \eqref{eq:maxwell}. See also \cite{israel:bivector:book}, where the conservation law is generalised to spacetimes of Petrov type II. The result for type D backgrounds fit into the picture of Penrose potentials\cite{Goldberg:1990} and in the next section we will see that it generalizes to linear perturbations.

\section{Fackerell's conservation law} \label{sec:fack}

We can of course linearize the two-form \eqref{eq:maxwell}, which would
provide a charge for perturbations within the class of type D spacetimes. But
more generally, Fackerell \cite{fackerell:1982} derived a closed two-form for
arbitrary linear perturbations around a type D 
background\footnote{One can expect that such a structure for perturbations of
  algebraically special solutions exists also for other signatures. A
  classification of the Weyl tensor in Euclidean signature can be found in
\cite{Karlhede:1986}, see also 
  \cite{hacyan:1979PhLA...75...23H}, a unified formulation for arbitrary signature is given
  in \cite{Batista:2012}.}. 
Starting from this conservation law, Fackerell and Crossmann derived field equations for
perturbations of Kerr-Newmann spacetime. Let us give a shortened derivation
in the vacuum case.  

When linearizing (with parameter $\pert$) the general bivector equations around a type D background in principal tetrad, we have
\begin{align*}
\Gamma = \tau \NPl - \rho \NPm + O(\pert) && \Gamma' = \tau' \NPn - \rho' \NPmbar + O(\pert)
\end{align*}
and it follows
\begin{align*}
\d^\theta Z^0 = - \tfrac{1}{2} h \wedge Z^0 + O(\pert) &&
\d^\theta Z^2 = - \tfrac{1}{2} h \wedge Z^2 + O(\pert) \\
\Gamma' \wedge \Sigma_0 = (\tau' \NPm - \rho' \NPl) \wedge\Sigma_1 + O(\pert^2) && \Gamma \wedge \Sigma_2 = (-\tau \NPmbar + \rho \NPn) \wedge \Sigma_1 + O(\pert^2)\, .
\end{align*}
\begin{proof}
Since
\begin{align*}
\Sigma_0 &= \Psi_0 Z^0 + \Psi_1 Z^1 + \Psi_2 Z^2 \\
\Sigma_1 &= \Psi_1 Z^0 + \Psi_2 Z^1 + \Psi_3 Z^2 \\
\Sigma_2 &= \Psi_2 Z^0 + \Psi_3 Z^1 + \Psi_4 Z^2
\end{align*}
and
\begin{align*}
Z^0 = \NPmbar \wedge \NPn &&
Z^1 = \NPn \wedge \NPl - \NPmbar \wedge \NPm &&
Z^2 = \NPl \wedge \NPm
\end{align*}
we have
\begin{align*}
\Gamma' \wedge \Sigma_0 
&= (\tau' \NPn +\kappa \NPl - \rho' \NPmbar - \sigma \NPm) \wedge (\Psi_0 Z^0 + \Psi_1 Z^1 + \Psi_2 Z^2) \\
&= 
\Psi_0 (\kappa l \wedge \NPmbar \wedge n - \sigma m \wedge \NPmbar \wedge n) \\
&\hspace{5pt} - \Psi_1 (\rho' \NPmbar \wedge n \wedge l + \sigma \NPm \wedge n \wedge l + \tau' \NPn \wedge \NPmbar \wedge \NPm +\kappa \NPl \wedge \NPmbar \wedge \NPm) \\
&\hspace{5pt} + \Psi_2 (\tau' \NPn \wedge \NPl \wedge \NPm - \rho' \NPmbar \wedge \NPl \wedge \NPm) \\
&= - \Psi_1 (\rho' \NPmbar \wedge n \wedge l + \tau' \NPn \wedge \NPmbar \wedge \NPm) 
+ \Psi_2 (\tau' \NPn \wedge \NPl \wedge \NPm - \rho' \NPmbar \wedge \NPl \wedge \NPm) + O(\eps^2) \\
&= \Psi_1(-\rho' l + \tau' m) \wedge Z^0 + \Psi_2 ( \tau' m - \rho' l) \wedge Z^1 + O(\eps^2)
\end{align*}
Because $ (\tau' m - \rho' l) \wedge Z^2 = 0$ this could be added which yields the result.
\end{proof}

Now expanding the Bianchi identitiy \eqref{eq:bianchi}, we find $ \d \Sigma_1 = \thalf h \wedge \Sigma_1 + O(\pert^2) $ which can be written
\begin{align} 
(\d - \tfrac{1}{2}h \wedge) \Sigma_1 = O(\pert^2)
\end{align}
In the background, this gives the Jordan-Ehlers-Sachs conservation law \eqref{eq:jes}. For linearized gravity, making use of $3 h \Psi_2 = 2\d \Psi_2$, we find the identity
\begin{equation}\label{eq:fakint}\begin{aligned}
0 &= \psi(\d - \tfrac{1}{2}h \wedge) \dot \Sigma_1 
- \tfrac{1}{2}\psi \dot h \wedge \Sigma_1 \\
&=\d ( \psi{\dot\Psi_1} Z^0 + \psi{\dot \Psi_2}Z^1 + \psi{\dot \Psi_3} Z^2 + \psi\Psi_2 \dot Z^1) -\tfrac{1}{2} \psi\Psi_2 \dot h \wedge Z^1 \\
&= \d ( \psi{\dot \Psi_1} Z^0 + \psi{\dot \Psi_2}Z^1 + \psi{\dot \Psi_3} Z^2 + \tfrac{3}{2} \psi\Psi_2 \dot Z^1) \, ,
\end{aligned}\end{equation}
were the linearized version of $\d Z^1 = -h \wedge Z^1$ is used in the last step. Note, that also
\begin{align}
0 = \d ( \psi{\dot \Psi_1} Z^0 + \psi{\dot \Psi_2}Z^1 + \psi{\dot \Psi_3} Z^2) -\tfrac{3}{2} \psi\Psi_2 \dot h \wedge Z^1
\end{align}
holds, which looks similar to Maxwell equations with a source. We summarize the above discussion by the following
\begin{thm}
For linearized gravity on a vacuum type D background in principal tetrad exists a closed two-form 
\begin{align} \label{eq:linmass}
 \linmass = \psi\dot\Psi_1 Z^0 + \psi\dot\Psi_2 Z^1 + \psi\dot\Psi_3 Z^2 + \tfrac{3}{2} \psi\Psi_2 \dot Z^1
\end{align}
which can be used to calculate the ``linearized mass''. The integral
\begin{align} \label{eq:1stintcond}
 \frac{1}{4\pi \i} \int_{S^2} \linmass 
\end{align}
is conserved, gauge invariant and gives the linearized ADM mass.
\end{thm}
The gauge invariance follows already from its relation to the ADM mass, but the integrand itself has interesting behaviour under gauge transformations. Beside infinitesimal changes of coordinates (coordinate gauge), there are infinitesimal Lorentz transformations of the tetrad (tetrad gauge). To discuss the second one, we need some notation. Following \cite{Crossmann:1976}, introduce 4 real functions $N_1,N_2,L_1,L_2 $ and 6 complex functions $L_3,N_3,M_i, i=1,..,4 $ to relate the linearized tetrad to the background tetrad
\begin{align} \label{eq:perttetrad}
 \begin{pmatrix}\NPl^a \\ \NPn^a \\ \NPm^a \\ \ba \NPm^a \end{pmatrix}_B =
\begin{pmatrix} 
L_1 & L_2 & L_3 & \ba L_3 \\
N_1 & N_2 & N_3 & \ba N_3 \\
M_1 & M_2  & M_3 & M_4\\
\ba M_1 & \ba M_2  & \ba M_4 & \ba M_3
\end{pmatrix}_B 
\begin{pmatrix}
 \NPl^a \\ \NPn^a \\ \NPm^a \\ \bar \NPm^a
\end{pmatrix} .
\end{align}
These are 16 d.o.f. at a point, 10 correspond to metric perturbations and 6 are infinitesimal Lorentz transformations (tetrad gauge). The linearized tetrad one-forms have the representation
\begin{align} 
 \begin{pmatrix}\NPl_a \\ \NPn_a \\ \NPm_a \\ \ba \NPm_a \end{pmatrix}_B =
\begin{pmatrix} 
-N_2 & - L_2 & \ba M_2 &  M_2 \\
-N_1 & -L_1 & \ba M_1 &  M_1 \\
\ba N_3 & \ba L_3  & - \ba M_3 & - M_4 \\
N_3 & L_3  & - \ba M_4 & - M_3
\end{pmatrix}_B 
\begin{pmatrix}
 \NPl_a \\ \NPn_a \\ \NPm_a \\ \ba \NPm_a
\end{pmatrix}.
\end{align}
It follows
\begin{align}
 \dot Z^0 &= -(L_1+M_3)Z^0 + \thalf(\ba M_1+N_3)Z^1 -\ba M_4\ba Z^0 - \thalf (\ba M_1-N_3)\ba Z^1 +N_1 \ba Z^2 \nonumber \\
\dot Z^1 &=  -(M_2 +\ba L_3)Z^0 - \thalf(L_1+N_2+M_3+\ba M_3)Z^1 -(\ba M_1 +N_3)Z^2 \nonumber\\ & \hspace{4mm} +(L_3-\ba M_2)\ba Z^0 - \thalf(L_1+N_2-M_3-\ba M_3)\ba Z^1 +(\ba N_3 -M_1)\ba Z^2 \label{eq:linZ1} \\
\dot Z^2 &= - \thalf(M_2 +\ba L_3)Z^1 -(N_2 +\ba M_3)Z^2 +L_2 \ba Z^0 + \thalf(M_2-\ba L_3)\ba Z^1 -M_4 \ba Z^2. \nonumber
\end{align}
Linearization of the tetrad representation of the metric yields
\begin{align*}
 h_{\NPl\NPn}  = -L_1 - N_2 &&&   
h_{\NPm\NPmbar}  = \ba M_3 + M_3 &&&
h_{\NPn\NPmbar} = N_3 -\ba M_1 &&& 
h_{\NPl\NPm} = \ba L_3 - M_2 
\end{align*}
and therefore $\textrm{tr}_g h = - 2 ( L_1 + N_2 + M_3 + \ba M_3 )$. One should also note, that the selfdual components of $\dot Z^1$ in $\linmass$ cancel some of the additional terms, not coming from the linearized Weyl tensor,
\begin{subequations}\begin{align}
 \dot \Psi_1 &= - \dot C \cdot (Z_0,Z_1) + \tfrac{3}{2} (\ba L_3 + M_2) \Psi_2  \label{eq:psi1b}\\
 \dot \Psi_2 &=- \dot C \cdot (Z_1,Z_1) + (L_1 + N_2 + M_3 + \ba M_3) \Psi_2 \label{eq:psi2b}\\
 \dot \Psi_3 &= - \dot C \cdot (Z_2,Z_1) + \tfrac{3}{2} (N_3 + \ba M_1) \Psi_2 \, . \label{eq:psi3b} 
\end{align}\end{subequations}

Using these facts, we show
\begin{prop}\label{prop:exact}
  On a spacetime of Petrov type D, the two-form $\linmass$ is tetrad gauge invariant and changes only with a term $\chi$ which is exact, $\chi=\d f$, under coordinate gauge transformations.
\end{prop}
\begin{remark} In the work of Fayos et al. \cite{fayos:ferrando:jaen:1990},
  a gauge in which $\d \left( \psi \Psi_2\dot Z^1 \right) = 0$ was used. It
  is not clear from that work 
  whether this gauge condition 
is compatible with a hyperbolic system of evolution
  equations for linearized gravity. 
\end{remark}
\begin{proof}[Proof of Proposition \ref{prop:exact}]
Let us first look at the coordinate gauge. Under infinitesimal coordinate transformations $x^a \to x^a + \xi^a$, a tensor field transforms with Lie derivative, $T \to T - \lie_\xi T$. For linearized gravity, we write this as $\dot T \to \dot T + \delta \dot T = \dot T - \lie_\xi T$. Now look at the middle bivector component $Z^1$ and use Cartan's identity $ \lie_\xi \omega = \d (\xi \invneg \omega) + \xi \invneg \d\omega $, which holds for arbitrary forms $\omega$. It follows for coordinate gauge transformations in $\linmass$,
\begin{equation}\label{eq:masscoordgauge}\begin{aligned}
 \delta \linmass
&= -\psi\xi(\Psi_2) Z^1 - \tfrac{3}{2}\psi\Psi_2 [ \d(\xi \invneg Z^1) +  \xi \invneg \d Z^1 ] \\
&= - \tfrac{3}{2}\psi\Psi_2 (\d + h\wedge)(\xi \invneg Z^1) \\
&= - \tfrac{3}{2} \d [ \psi\Psi_2(\xi \invneg Z^1)]
\end{aligned}\end{equation}
where $\xi \invneg h = \tfrac{2}{3} \Psi_2^{-1} \xi(\Psi_2)$ and $\xi \invneg (h\wedge Z^1) = (\xi \invneg h) Z^1 - h \wedge (\xi \invneg Z^1)$ was used. The two-form \eqref{eq:masscoordgauge} is exact and hence integrates to zero.

A tetrad gauge transformation changes the tetrad \eqref{eq:perttetrad} as follows,
\begin{align} \label{eq:tetradgauge}
 \delta\begin{pmatrix}\NPl^a \\ \NPn^a \\ \NPm^a \\ \NPmbar^a \end{pmatrix}_B =
\begin{pmatrix} 
A & 0 & \bar b & b \\
0 & -A & \bar a & a \\
a & b & \i \vartheta & 0 \\
\bar a & \bar b & 0 & -\i \vartheta
\end{pmatrix} 
\begin{pmatrix}
 \NPl^a \\ \NPn^a \\ \NPm^a \\  \NPmbar^a
\end{pmatrix}_B
\end{align}
with $a,b$ complex and $A,\vartheta$ real valued. It follows, that the tetrad gauge dependent terms in (\ref{eq:psi1b},\ref{eq:psi3b}) cancel the ones in \eqref{eq:linZ1}. The anti selfdual part in \eqref{eq:linZ1} is invariant, as follows from \eqref{eq:tetradgauge}. This shows the tetrad gauge invariance of $\linmass$ and therefore gauge invariance of \eqref{eq:1stintcond}. 
\end{proof}

Finally, to express the charge integral in a form similar to the Maxwell case \eqref{eq:integrability}, we need the $\theta\phi$ components of the bivectors,
\begin{align}
 Z^1_{\theta\phi}=-\i (r^2+a^2)\sin\theta &&& Z^0_{\theta\phi} = -Z^2_{\theta\phi} = \frac{a \sqrt{\Delta}}{2} \sin^2 \theta.
\end{align}
The charge integral becomes
\begin{align}
\frac{2\i}{\sqrt{\Delta}} \int_{S^2(t,r)} \mathcal{\dot M} =
\int_{S^2(t,r)} \left( 2 V_L^{-1/2} \dot{\widehat \Psi}_2 + \i a \sin \theta \Psi_{diff} \right) (r-\i a \cos\theta) \d\mu
\end{align}
with $V_L = \Delta/(r^2+a^2)^2$, $\d\mu=\sin\theta\d\theta\d\varphi$ and
\begin{subequations}\begin{align}
  \dot{\widehat \Psi}_2  &= \dot\Psi_2 -\Psi_2(M_3 + \bar M_3) \\
 \Psi_{diff} &= \dot\Psi_1 - \dot\Psi_3 - 3\Psi_2 \left( \Re(M_2 - M_1) - \i \Im(L_3 + N_3) \right)  \, .
\end{align}\end{subequations}

\section{Conclusions} \label{sec:conclusions} 
For each isometry of a given background, there is a conserved charge for the linearized gravitational field. Working in terms of linearized curvature, we derived a linearized mass charge (corresponding to the time translation isometry) for Petrov type D backgrounds, by using Penrose's idea of spin-lowering with a Killing spinor.

A second Killing spinor, corresponding to the axial isometry of Kerr spacetime does not exist, \eqref{eq:typeDKS}. Hence spin lowering cannot be used directly to derive a linearized angular momentum charge, even tough a canonical analysis provides one in terms of the linarized metric.

For a Schwarzschild background, gauge conditions are known, which eliminate
the gauge dependent non-radiating modes
\cite{zerilli:1970,jezierski:1999}. Understanding these conditions in a
geometric way and generalizing them to a Kerr background needs further investigation.

\appendix
\section{Coordinate expressions}
Using a Carter tetrad, the bivectors and connection one-forms in Boyer-Lindquist coordinates are
\begin{subequations}\begin{align}
 Z^1_{ab} &= 
\begin{pmatrix} 
 0 & -1 & -\i a \sin\theta & 0 \\
1 & 0 & 0 & -a \sin^2\theta \\
\i a \sin\theta & 0 & 0 & -\i (r^2+a^2)\sin\theta \\
0 & a \sin^2\theta & \i (r^2+a^2)\sin\theta & 0
\end{pmatrix} \\
 Z^0_{ab} &= \frac{1}{2\sqrt{\Delta}}
\begin{pmatrix} 
 0 & - \i a \sin\theta & \Delta & -\i \Delta \sin\theta \\
 \i a \sin\theta & 0  &\Sigma & - \i (r^2+a^2)\sin \theta \\
-\Delta & - \Sigma & 0 & a \Delta \sin^2 \theta  \\
\i \Delta \sin\theta  & \i (r^2+a^2)\sin \theta & -a \Delta \sin^2 \theta  & 0
\end{pmatrix} \\
 Z^2_{ab} &= \frac{1}{2\sqrt{\Delta}}
\begin{pmatrix} 
 0 &  \i a \sin\theta & -\Delta & -\i \Delta \sin\theta \\
 -\i a \sin\theta & 0  &\Sigma &  \i (r^2+a^2)\sin\theta \\
\Delta & - \Sigma & 0 & -a \Delta \sin^2\theta  \\
\i \Delta \sin\theta  & -\i (r^2+a^2)\sin\theta & a \Delta \sin^2 \theta  & 0
\end{pmatrix}
\end{align}\end{subequations}
\begin{subequations}\begin{align}
\sigma_{0a} &= \left( 0 , \frac{\i a \sin\theta}{2 p \sqrt{\Delta}} , - \frac{\sqrt{\Delta}}{2p}, -\frac{\i \sqrt{\Delta} \sin\theta}{2p} \right) \\
\sigma_{1a} &= \left( \frac{M}{2p \Sigma} , 0 , 0 , -\frac{Ma\bar p^2 \sin^2 \theta + ra \sin^2\theta \Sigma + \i \cos\theta (r^2+a^2)\Sigma}{2\Sigma^2}  \right) \\
\sigma_{2a} &= \sigma_{0a}
\end{align}\end{subequations}
Here, we used
\begin{align*}
p = r - \i a \cos\theta, && \Sigma = p \bar p, && \Delta = r^2 - 2Mr + a^2
\end{align*}

\subsection*{Acknowledgements} We would like to thank
Ji{\v{r}}{\'{\i}} Bi{\v{c}}{\'a}k, 
Pieter Blue, Domenico Giulini, Jacek Jezierski, Lionel Mason and 
Jean-Philippe Nicolas for helpful discussions.  
One of the authors (S.A.) gratefully acknowledges the support of
the \emph{Centre for Quantum Engineering and Space-Time Research (QUEST)} and
the \emph{Center of Applied Space Technology and Microgravity
  (ZARM)}, and thanks the Albert Einstein Institute, Potsdam for hospitality
during part of the work on this paper. 
L.A. thanks the Department of Mathematics of the University of
Miami for hospitality during part of the work on this paper. 

\newcommand{\prd}{Phys. Rev. D} 


\end{document}